\documentclass[showpacs,amsmath,amssymb,twocolumn,pra]{revtex4-1}
\usepackage[dvips]{graphicx} % for figures
\usepackage{amsfonts}
\usepackage{amssymb}
\usepackage{amscd}
\usepackage{amsmath}    % need for subequations
\usepackage{enumerate}
\usepackage{epsfig}
\usepackage{subfigure}
\usepackage{xcolor}
\usepackage{amsthm}
\usepackage{framed}
\usepackage[normalem]{ulem}

\newtheorem{lemma}{Lemma}

\newcommand{\bra}[1]{\mbox{$\left\langle #1 \right|$}}
\newcommand{\ket}[1]{\mbox{$\left| #1 \right\rangle$}}

\begin{document}
\preprint{APS/123-QED}
\title{Intrinsic randomness as a measure of quantum coherence}
\date{\today}
\author{Xiao Yuan}
\author{Hongyi Zhou}
\author{Zhu Cao}
\author{Xiongfeng Ma}
\affiliation{Center for Quantum Information, Institute for Interdisciplinary Information Sciences, Tsinghua University, Beijing, China}

\begin{abstract}
%Since the inception of quantum mechanics, the question of whether there exists intrinsic randomness in measurement outcomes has been long debated. Based on the theory of quantum mechanics, intrinsic randomness could distinguish quantum effects from classical ones. Hence, we would expect that such intrinsic randomness would be able to quantify the strength of `quantumness', or more specifically, quantum coherence. With such a motivation, we propose a coherence measure that quantifies the intrinsic randomness of the measurement outcomes of a quantum state. The classical randomness that can be predicted is removed from our coherence measure, thus leaving us with only truly random outcomes. Moreover, the randomness extraction can be replaced by coherence distillation, highlighting the resource perspective of quantum coherence.

Based on the theory of quantum mechanics, intrinsic randomness in measurement distinguishes quantum effects from classical ones. From the perspective of states, this quantum feature can be summarized as coherence or superposition in a specific (classical) computational basis. Recently, by regarding coherence as a physical resource, Baumgratz et al.~present a comprehensive framework for coherence measures. Here, we propose a quantum coherence measure essentially using the intrinsic randomness of measurement. The proposed coherence measure provides an answer to the open question in completing the resource theory of coherence. Meanwhile, we show that the coherence distillation process can be treated as quantum extraction, which can be regarded as an equivalent process of classical random number extraction. From this viewpoint, the proposed coherence measure also clarifies the operational aspect of quantum coherence. Finally, our results indicate a strong similarity between two types of quantumness --- coherence and entanglement.

%Here, we propose a coherence measure that essentially employs the intrinsic randomness of quantum measurement. In general, measurement processes contain classical and quantum ones, which generate classical and quantum randomness, respectively. The classical randomness that can be predicted is removed from our coherence measure, thus leaving us with only truly random outcomes. Moreover, the randomness extraction can be replaced by coherence distillation, highlighting the resource perspective of quantum coherence.

\end{abstract}
\maketitle

\section{introduction}
As one of the fundamental laws of quantum mechanics, Born's rule \cite{born1926quantentheorie} endows the real world with true randomness that does not exist in the classical Newtonian theory. Such is the counter-intuitiveness of the result that Einstein was quoted as saying `God does not play dice'. Nevertheless, the intrinsically random nature of  measurement outcomes is now considered a key characteristic that distinguishes quantum mechanics from classical theory \cite{bell}.

As a key feature of quantum mechanics,  coherence is often considered as a basic ingredient for quantum technologies \cite{giovannetti2011advances, lambert2013quantum}.
%Quantum coherence or superposition in a specific (classical) computational basis is often considered as a basic ingredient for quantum technologies \cite{giovannetti2011advances, lambert2013quantum}.
Considerable effort has been undertaken to theoretically formulate the quantum coherence \cite{Glauber63, Sudarshan, luo2005quantum, Aberg06, monras2013witnessing,Baumgratz14,Aberg14,Girolami14}. Recently, a comprehensive  framework of coherence quantification was established \cite{Baumgratz14}, by which coherence is considered to be a resource that can be characterized, quantified, and manipulated in a manner similar to that of another important feature--- quantum entanglement \cite{Bennett96, Vedral98, Plbnio07, Horodecki09}. Within the resource framework of coherence, several coherence measures are proposed based on relative entropy, $l_1$-norm \cite{Baumgratz14}, and skew-information \cite{Girolami14}. A thorough understanding of the resource theory of coherence is left as an interesting open question \cite{Baumgratz14}.

In measurement theory, decoherence, breaking coherence or superposition, in a specific (classical) computational basis results in random outcomes \cite{Zurek09}. Intuitively, from the resource perspective, randomness can be generated by consuming coherence of a quantum state. In order to quantitatively establish this connection, one needs to find a proper way to assess the randomness of measurement, which normally contains quantum and classical processes. The superficially random outcomes in classical processes are generally not truly random, although they might appear so if information is ignored. Thus, such classical part of randomness should be precluded when quantifying a quantum feature --- coherence. A quantum process, on the other hand, can generate genuine randomness, which we call intrinsic (quantum) randomness. Observing such intrinsic random outcomes of measurements would indicate non-classical (quantum) features of objects.

As an example, let us consider the famous Schr\"odinger's cat gedanken experiment as shown in Fig.~\ref{Fig:cat}. In a classical world, a cat might be either alive or dead  before observation, which can be described by the density matrix $\rho_{\mathrm{cat}}^\mathrm{C} = \left(\ket{\mathrm{live}}\bra{\mathrm{live}} + \ket{\mathrm{dead}}\bra{\mathrm{dead}}\right)/2$ for the case of being alive and dead equally likely, see Fig.~\ref{Fig:cat}~(a). The observation result of whether the cat is alive or dead looks random, which is due to the lack of knowledge of the cat system. After considering some hidden variables or an ancillary system $E$ that purifies $\rho_{\mathrm{cat}}^\mathrm{C}$, $\ket{\Psi} = \left(\ket{\mathrm{live}}\ket{0}_E + \ket{\mathrm{dead}}\ket{1}_E\right)/\sqrt{2}$, we can simply observe the system $E$ to infer whether the cat is alive or dead. In quantum mechanics, the cat can be in a coherent superposition of the states of alive and dead, $\rho_{\mathrm{cat}}^\mathrm{Q} = \ket{\psi}\bra{\psi}$, where $\ket{\psi} = \left(\ket{\mathrm{live}} + \ket{\mathrm{dead}}\right)/\sqrt{2}$, see Fig.~\ref{Fig:cat}~(b). The observation outcome would be intrinsically random according to Born's rule. That is, without directly accessing the system of the cat and breaking the coherence, we can never predict whether the cat is alive or dead better by blindly guessing.
%Intuitively, the existence of inherent randomness distinguishes a quantum state $\rho_{\mathrm{cat}}^Q$ from a classical one $\rho_{\mathrm{cat}}^\mathrm{C}$.
Therefore, the existence of intrinsic randomness can be regarded as a witness for quantum coherence.

\begin{figure}[hbt]
\centering \resizebox{6cm}{!}{\includegraphics{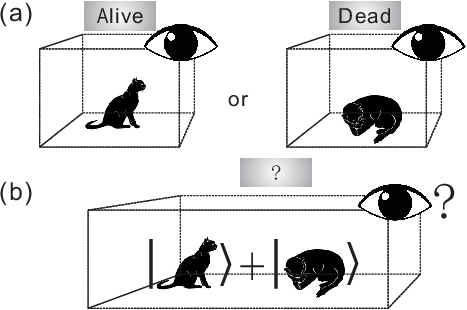}}
\caption{Illustration of Schr\"odinger's cat gedanken experiment.} \label{Fig:cat}
\end{figure}

With such strong evidence of the connection between coherence and intrinsic randomness, a natural question is whether we can consider the intrinsic randomness as a measure of coherence. If this is possible, production of a certain amount of intrinsic randomness will inevitably cause consumption of the same amount of coherence.

In this study, we explicitly answer this question by first proposing a coherence measure using intrinsic randomness and thus show the equivalence of the definitions between intrinsic randomness and quantum coherence. Then, we present a coherence distillation protocol for pure states and show that it is equivalent to random number extraction. Our distillation protocol provides an operational meaning to coherence, thus it answers the open question, stated in the literature \cite{Baumgratz14}, on the resource aspect of quantum coherence. Next, by noticing the similarity to the entanglement of formation (EOF) \cite{Bennett96, Hill97, Wootters98}, we provide an explicit way to evaluate our coherence measure for the qubit case. It is worth mentioning that the proposed measure is the first convex roof measure for coherence. Finally, we compare the two quantumness measures, coherence and entanglement, in a more general scenario.

\section{Coherence measures}
We first briefly review the framework of coherence measures \cite{Baumgratz14}.
The following discussion is focused on a general $d$-dimensional Hilbert space, if not specified. For a classical computational basis $I= \{\ket{i}\}_{i=1,2,\dots,d}$, which is similar to the alive and dead basis of the cat, quantum coherence can be interpreted as the superposition strength on the classical states from set $I$.
For example, any state that can be represented by a diagonal state of $I$, that is,
\begin{equation}\label{Eq:sigma}
  \delta = \sum_{i=1}^d p_i\ket{i}\bra{i},
\end{equation}
has no superposition, and is thus called an incoherent (classical) state. We label the set of such state by $\mathcal{I}$. Conversely, a maximally coherent state is given by the maximal superposition state
\begin{equation}\label{Eq:Psid}
\ket{\Psi_d} = \frac{1}{\sqrt{d}}\sum_{i = 1}^d \ket{i},
\end{equation}
up to arbitrary relative phases between the components $\ket{i}$.

Similar to the definition of local operations and classical communication (LOCC) in entanglement \cite{Bennett96, Vedral98, Plbnio07}, the incoherent operations are defined by incoherent completely positive trace preserving (ICPTP) maps $\Phi_{\mathrm{ICPTP}}(\rho) = \sum_n K_n\rho K_n^\dag$, where the Kraus operators $\{K_n\}$ satisfy $\sum_n K_n K_n^\dag = I$ and  $K_n \mathcal{I} K_n^\dag \subset \mathcal{I}$. For the case, where post-selections are enabled, the output state corresponding to the $n$th Kraus operation is given by $  \rho_n = {K_n\rho K_n^\dag}/{p_n}$, where $p_n = \mathrm{Tr}\left[ K_n\rho K_n^\dag\right]$ is the probability of obtaining the outcome $n$.

The amount of coherence can be quantified in a manner similar to entanglement \cite{Bennett96, Vedral98, Plbnio07}.
Generally, a measure of coherence is a map $C$ from quantum state $\rho$ to a real non-negative number that satisfies the properties listed in Table~\ref{Fig:Properties}.
\begin{table}[htb]
\begin{framed}
\centering
\begin{enumerate}[(C1)]
\item
Coherence vanishes for all incoherent state. That is, $C(\delta) = 0$, for all $\delta\in \mathcal{I}$. A stronger requirement claims that (C1') $C(\delta) = 0$, iff $\delta\in \mathcal{I}$.
\item
\emph{Monotonicity}: coherence should not increase under incoherent operations. Thus, (C2a) $C(\rho) \geq C(\Phi_{\mathrm{ICPTP}}(\rho))$, and (C2b) $C(\rho) \geq \sum_n p_nC(\rho_n)$, where (C2b) is for the case where post-selection is enabled.
\item
\emph{Convexity}: coherence cannot increase under mixing states, $\sum_e p_eC(\rho_e) \geq C(\sum_e p_e\rho_e)$.
\end{enumerate}
\end{framed}
\caption{Properties that a coherence measure should satisfy.} \label{Fig:Properties}
\end{table}
Based on the distance measure, coherence can be quantified by the minimum distance from $\rho$ to all the incoherent states in $I$ \cite{Baumgratz14}. Two examples are, respectively, based on the relative entropy
\begin{equation}\label{Eq:Crelent}
C_{\mathrm{rel, ent}}(\rho) \equiv \min_{\delta\in I}{S(\rho||\delta)},
\end{equation}
and the $l_1$ matrix norm
\begin{equation}\label{Eq:l1norm}
C_{l_1}(\rho) \equiv \min_{\delta\in I}\parallel\rho - \delta\parallel_{l_1} = \sum_{i\neq j}|\bra{i}\rho\ket{j}|.
\end{equation}

\section{Intrinsic randomness}
In quantifying the intrinsic randomness of measurement, we restrict on projective measurements $P_{I} = \{P_i = \ket{i}\bra{i}\}$ in the same classical basis $I$ \footnote{The definitions of coherence and intrinsic randomness are based on a specific computational basis. In this perspective, the quantum feature can be quantified by the superposition strength on the measurement basis. Alternatively, we can define similar quantumness as the ability of measurements. For an arbitrary pure quantum state, if we can choose the measurement basis that is complementary to the state, a quantum feature similar to coherence can also be maximally revealed. The definitions of coherence based on the property of quantum state with a given measurement basis and the property of measurement is similar to the relationship between the pictures of Schrodinger and Heisenberg. The current definition of coherence thus follows from the routine of the Schrodinger¡¯s picture.}. Here, we define intrinsic randomness as the random outcomes that can not be predicted.

For example, when measuring a pure state $\rho = \ket{\psi}\bra{\psi}$, where $\ket{\psi} = \sum_ia_i\ket{i}$,  the measurement outcomes are truly random according to Born's rule.  Let $p_i = \mathrm{Tr}[P_i\rho] = |a_i|^2$ be the probability of obtaining the $i$th outcome, the randomness of the output random variable $A$ can be quantified by
%\begin{equation}\label{Eq:Randompure}
  $R_I(\ket{\psi}\bra{\psi}) = H(A) \equiv-\sum_ip_i\log_2 p_i$,
%\end{equation}
where $H$ is the Shannon entropy function on the probability distribution $\{p_i\}$. Define $\rho^{\mathrm{diag}}$ to be the density matrix that has only diagonal terms of $\rho$ in the computational basis $I$. We can rewrite $R_I(\ket{\psi}\bra{\psi})$ as
\begin{equation}\label{Eq:puredef}
  R_I(\ket{\psi}\bra{\psi}) = S(\rho^{\mathrm{diag}})
\end{equation}
where $S$ is the von Neumann entropy function.
Suppose that the projective measurement is performed on $N$ copies of  $\ket{\psi}$, it is evident that the $N$ outcomes are independent and identically distributed (i.i.d.) random variables. With the Shannon source coding theorem \cite{shannon2001mathematical}, these random outcomes can be compressed into about $NH(A)$ bits, thus intuitively explaining why $H(A)$ quantifies the average randomness of the measurement outcome. We emphasize that our results can also be derived with other entropy functions \cite{rrnyi1961measures}, such as min-entropy, which is also widely used to quantify randomness. Here, we only consider the case where the measurement outcomes are i.i.d.~and leave the general case in Appendix \ref{Appendix:general}.

For a general mixed quantum state $\rho$, one might naively quantify the randomness in a similar manner to the pure state case.
%By first calculating the probability distribution $p_i = \mathrm{Tr}[P_i\rho]$ of the $i$th outcome, the randomness can be defined according to Eq.~\eqref{Eq:Randompure}.
Clearly, this definition overestimates the intrinsic randomness. For instance, consider a maximally entangled bipartite state $\ket{\psi^{AE}}= \left(\ket{00}+\ket{11}\right)/2$ shared by Alice and Eve. Alice performs projection measurements on her quantum states to gain random numbers. Suppose that the measurement basis is $I = \{\ket{0},\ket{1}\}$, Alice's outputs look random, but they can always be predicted by Eve, who simply measures her qubits on the same basis. Equivalently, the system $E$ can be regarded as a hidden variable that determines the state of Alice with certainty. Therefore, we should not recognize this type of randomness as being intrinsic randomness.

\begin{figure}[hbt]
\centering \resizebox{6cm}{!}{\includegraphics{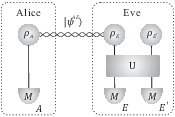}}
\caption{Alice performs projection measurement in the $I$ basis on a quantum state $\rho_A$, which could possibly be entangled with $\rho_E$. %The inherent measurement randomness is defined by the randomness of the output random variable $A$ of Alice conditioned on all other possible predictions $E$ and $E'$.
} \label{Fig:Channel}
\end{figure}

Instead, we consider a purified state, $\ket{\psi^{AE}}$, that is shared by Alice and an adversary, Eve, who attempts to predict the outputs of Alice's measurement as shown in Fig.~\ref{Fig:Channel}.
The intrinsic randomness quantifies the randomness of Alice's measurement outcomes $A$, conditioned on Eve's predictions $E$ and $E'$.
As the operations of Alice and Eve commute with each other, we can, without loss of generality, imagine that Eve performs her measurement first. For simplicity, suppose that Eve performs the projection measurement $\left\{\ket{\psi_e^E}\bra{\psi_e^E}\right\}$ on her state. When Eve obtains an outcome $e$ with probability $p_e$, the state of Alice is $\ket{\psi_e^A} = \langle{\psi_e^E}\ket{\psi^{AE}}$. As we already know, the measurement randomness that Alice can generate on $\ket{\psi_e^A}$ is given by $R_I(\ket{\psi_e^A})$. The total randomness can be quantified by $\sum p_e R_I\left(\ket{\psi_e^A}\right)$, where $\rho_A = \sum_ep_j\ket{\psi_e^A}\bra{\psi_e^A}$. As Eve could choose her measurement basis to maximize the probability of guessing Alice's measurement outcome, the intrinsic randomness that Alice can generate should take the minimum of all possible decompositions of $\rho_A$, that is,
\begin{equation}\label{Eq:generaldef}
  R_I\left(\rho\right) = \min_{\{p_e,\ket{\psi_e}\}}\sum_e p_e R_I(\ket{\psi_e}),
\end{equation}
where $\rho = \sum_e p_e \ket{\psi_e}\bra{\psi_e}$ and $\sum_{e}p_e= 1$. Notice that, as the minimization runs over all possible decomposition, the definition $R_I\left(\rho\right)$ does not depend on the purification.

For the case when Eve performs general positive-operator valued measures (POVMs), we can first `purify' the measurement and consider projection measurement on a quantum state in a larger Hilbert space as $\ket{\psi^{AEE'}}$, where Alice has $A$ and Eve has $EE'$.
Therefore a similar proof for POVMs follows.

\section{Verifying the properties of $R$}
Now we show that the intrinsic randomness $R_I$, defined in Eq.~\eqref{Eq:generaldef}, satisfies the properties of coherence measure listed in Table~\ref{Fig:Properties}. That is, the requirements of the measures for quantum coherence and intrinsic randomness are equivalent.

In the language of generating randomness, the requirement (C1) in Table~\ref{Fig:Properties} can be interpreted as saying classical states generate no randomness. This is because an incoherent state $\delta$, defined in Eq.~\eqref{Eq:sigma}, can be understood as a statistical mixture of classical states. We can easily verify that $R_I(\delta) = 0$, since $R_I(\delta) \leq  \sum_{i=1}^d p_iR_I(\ket{i}\bra{i}) = 0$ from Eq.~\eqref{Eq:sigma} and $R_I(\rho) \geq 0$ by definition. The stronger requirement (C1') implies that any non-classical states, which cannot be represented in the form of Eq.~\eqref{Eq:sigma}, could always be used to generate intrinsic randomness. Thus, this result answers why nonzero intrinsic randomness always indicates `quantumness' as discussed above. The proof for (C1') is provided in Appendix \ref{Appendix:proof}. We can also show that the upper bound of its intrinsic randomness is given by $R_I\left(\rho\right) \leq \log_2d$. The maximally coherent state $\ket{\Psi_d}$, defined in Eq.~\eqref{Eq:Psid}, has the largest intrinsic randomness.

The requirement (C2) implies a monotonicity property of incoherent operations. In the corresponding randomness picture, incoherent operations can be understood as classical operations that map one zero intrinsic randomness (classical) state to another one. An interpretation of (C2a) is that such classical operations should not increase randomness of a given state. While (C2b) requires that the randomness cannot increase on average when probabilistic strategies are considered. Let us quickly check why (C2b) is true for the pure state case, while leaving the proof for other cases in Appendix \ref{Appendix:proof}. For a pure state $\rho$, the randomness measure $R_I(\rho)$ equals the relative entropy of coherence $C_{\mathrm{rel, ent}}(\rho)$, whose monotonicity has been proved \cite{Baumgratz14}.

The convexity property (C3) can be understood as a requirement on the randomness generation process. In other words, the randomness cannot increase on average by statistically mixing several states. With the convex roof definition of $R_I(\rho)$, given in Eq.~\eqref{Eq:generaldef}, we can easily verify the convexity property (C3). The proof follows directly by considering a specific decomposition of $\rho = \sum_n p_n\rho_n$ in (C3).
Note that, the property (C2a) can be derived when (C2b) and (C3) are fulfilled, thus we also prove (C2a) for $R_I(\rho)$.

In summary, we prove that the intrinsic randomness $R_I(\rho)$ indeed measures the strength of coherence. A state with stronger coherence would therefore indicate larger randomness in measurement outcomes, and vice versa.

\section{Randomness distillation}
As mentioned above, when Alice performs a projective measurement $P_I$ on $N$ identical pure states $\ket{\psi} = \sum_i a_i\ket{i}$, she will obtain $N$ i.i.d.~random variables $A_1, A_2, \dots, A_N$. For the state $\ket{\psi}$ that is not maximally coherent, the randomness of the measurement outcomes is biased. Then, as shown in Fig.~\ref{Fig:Extractor}(a), Alice can perform a randomness extraction process to transform the $N$ biased random numbers to $l \approx NR_I(\ket{\psi})$ almost uniformly distributed random bits.

\begin{figure}[hbt]
\centering \resizebox{8cm}{!}{\includegraphics{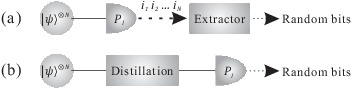}}
\caption{Random number extraction and coherence distillation. The randomness extraction process can be replicated by first distilling the coherence of the quantum state. Measurement outcomes will directly produce uniformly random bits.} \label{Fig:Extractor}
\end{figure}

We show in Fig.~\ref{Fig:Extractor}(b) that extraction can be equivalently performed before measurement. Now, extraction becomes a quantum procedure, which we call \emph{quantum extraction}. Considering the equivalence between intrinsic randomness and quantum coherence, quantum extraction can be regarded as a procedure of \emph{coherence distillation}. This concept resembles the distillation procedure of another (more popular) quantumness measure---entanglement \cite{Bennett96}.

With quantum extraction, we can first distill the input state $\ket{\psi} = \sum_i a_i\ket{i}$ into the maximally coherent state $\ket{\Psi_2} = (\ket{0} + \ket{1})/\sqrt{2}$. Then, we can directly obtain uniformly distributed random bits by measuring the maximally coherent state. For $N$ copies of $\ket{\psi}$, it is shown in Appendix \ref{App:distill} that we can asymptotically obtain $l$ copies of $\ket{\Psi_2}$, where $l$ and $N$ satisfy the following condition,
\begin{equation}\label{Eq:CohDistill}
{l}/{N} \approx R_I\left(\ket{\psi}\right).
\end{equation}
Taking a pure qubit input state as an example, the distillation procedure is summarized as follows.
\begin{enumerate}
\item
Prepare $N$ copies of qubit state $\ket{\psi}^{\otimes N} = \left(\alpha \ket{0} + \beta \ket{1}\right)^{\otimes N}$, which can be binomially expanded on the computational basis. There are $N+1$ distinct coefficients, $\beta^N, \alpha^{1}\beta^{N-1}, \dots, \alpha^N$, corresponding to different subspaces that have the same number of $\ket{0}$ or $\ket{1}$.

\item
Perform a projection measurement to distinguish between those subspaces. For the $k$th subspace, which has coefficient $\alpha^{N-k}\beta^{k}$, the measurement probability is given by $p_k = {N\choose k}|\alpha|^{2(N-k)}|\beta|^{2k}$. The resulting quantum state of the $k$th outcome corresponds to a maximally coherent state $\ket{\Psi_{D_k}}$  of dimension $D_k = {N\choose k}$.

\item
Suppose that $2^r\leq D_k< 2^{r+1}$, then we can directly project onto the $2^r$ subspace and convert to $r$ copies of $\ket{\Psi_2}$ as desired.
\end{enumerate}
%We refer to Supplementary Materials for detailed discussions.

To see why $r/N$ equals the randomness of $\ket{\psi}$ on average, we only need to take account of the operations that cause a loss of coherence. As shown in Appendix \ref{App:distill}, the only two projection measurements lose negligible amount of coherence, thus we asymptotically have $NR_I(\ket{\psi}) \approx r$.

In Appendix \ref{Appendix:general}, we further extend the definition of distillable coherence to mixed quantum states. Compared to the definition of the regularized entanglement of formation \cite{hayden2001asymptotic}, we also define coherence of formation and conjecture that it equals the regularized intrinsic randomness measure,
\begin{equation}\label{Eq:RNrho}
  R_I^\infty(\rho) = \lim_{N\rightarrow\infty}\frac{ R_I\left(\rho^{\otimes N}\right)}{N}.
\end{equation}
%In our distillation procedure, the two projection measurements might lose coherence. As shown in Supplementary Materials, both inherent operations only lose negligible coherence, and hence we have an asymptotically relation $NR_I(\ket{\psi}) \approx r$.

\section{Qubit example}
Here, we give an example of the calculation of $R_I(\rho)$ for a qubit state $\rho$. We follow a method of deriving the EOF \cite{Bennett96, Hill97,Wootters98} and refer to Appendix \ref{app:qubit} for details. Denote the Pauli matrices by $\sigma_i, \sigma_x, \sigma_y$, and $\sigma_z$.  When measured in the $\sigma_z$ basis, the randomness $R_{z}\left(\rho\right)$ can be calculated by
\begin{equation}\label{Eq:ef}
R_{z}(\rho)= H\left(\frac{1+\sqrt{1-C_z^2}}{2}\right).
\end{equation}
Here, the  $C_z$ term is defined as $ C_z =  |\sqrt{\eta_1}-\sqrt{\eta_2}|$, which resembles the concurrence \cite{Hill97}, where $\eta_1$ and $\eta_2$ are the eigenvalues of the matrix $M = \rho \sigma_x \rho^* \sigma_x$. In the Bloch sphere representation, the value of $C_z$ of a quantum state $\rho = (\sigma_i + n_x \sigma_x + n_y\sigma_y +n_z \sigma_z )/2$ can be calculated by
%\begin{equation}\label{Eq:concurrence}
  $C_z = \sqrt{n_x^2 + n_y^2}$.
%\end{equation}
Compared with the $l_1$ norm coherence measure $C_{l_1}$ defined in Eq.~\eqref{Eq:l1norm},  we can easily check that $C_{l_1}(\rho)$ equals the coherence concurrence $C_z$ for the qubit case. We conjecture that the coherence concurrence can be generalized to an arbitrary high-dimensional space by following a similar method to that used for the entanglement concurrence \cite{Rungta01, Audenaert01, badziag2002concurrence}.

\section{Discussion}
As shown in Table~\ref{coherenceEntanglement}, there exist strong similarities between the frameworks of coherence and entanglement (see also Ref.~\cite{Streltsov15}), our study can be regarded as an extension of the convex roof measure from entanglement to coherence. Similar to the case of EOF, as a convex roof measure for coherence, we expect our proposed measure to play an important role in the research of coherence.

\begin{table*}[hbt]\centering
\caption{Comparing  the frameworks of coherence and entanglement. DI: device-independent; MDI: measurement-device-independent; QKD: quantum key distribution; QRNG: quantum random number generation.}\label{ps}
\begin{tabular}{ccc}
\hline
Properties & Coherence & Entanglement \\
\hline
Classical operation & Inherent operation \cite{Baumgratz14} & LOCC \cite{Bennett96}\\
Classical state & Incoherent state, Eq.~\eqref{Eq:sigma} & Separable state\\
Distance measure & $C_{\mathrm{rel, ent}}(\rho)$, Eq.~\eqref{Eq:Crelent} & Relative entropy distance \cite{Vedral98}\\
Convex roof measure & $R_I(\rho)$, Eq.~\eqref{Eq:generaldef} & EOF \cite{Bennett96, Hill97, Wootters98}\\
Distillation & Coherence distillation (Appendix \ref{App:distill}) & Entanglement distillation \cite{Rains98, Horodecki09}\\
Formation (cost) & Coherence formation & Entanglement cost \cite{hayden2001asymptotic, Horodecki09}\\
Foundation tests & Further research direction & Nonlocality tests \cite{bell, CHSH}  \\
Interconvertibility &  \cite{Du15a, Du15b} & Deterministic \cite{Nielsen99}, stochastic \cite{Dur00}\\
Catalysis effect & Further research direction & Entanglement catalysis \cite{Jonathan99, Eisert00}\\
Witness & Further research direction & Entanglement witness (EW) \\
DI applications & Further research direction & DIQKD \cite{Mayers98, Acin}, DIQRNG \cite{vazirani2011certifiable}\\
MDI applications & Further research direction & MDIQKD \cite{Lo12, Braunstein12}, MDIEW \cite{Branciard13, Xu14}\\
\hline
\end{tabular}\label{coherenceEntanglement}
\end{table*}

For further research directions, it is interesting to extend the framework of entanglement to coherence. An incomplete list of comparison between the two are shown in Table~\ref{coherenceEntanglement}. For instance, it is interesting to see whether $C_{\mathrm{rel, ent}}(\rho)$ and $R_I(\rho)$ are the unique lower and upper bounds of all coherence measures after regularization, and whether they can coincide.
Another interesting and related question is that of quantifying the coherence for an unknown quantum state, similar to the task of using an entanglement witness for quantification. The coherence measure $R_{I}(\rho)$ given in Eq.~\eqref{Eq:generaldef} ensures the true randomness when measuring a state $\rho$ in the $I$ basis. Such a technique can be utilized to construct a semi self-testing quantum random number generator. A straightforward way to do this is to first perform tomography on the to-be-measured state $\rho$ and then estimate the randomness of the $I$ basis measurement outcomes according to Eq.~\eqref{Eq:generaldef}. As the coherence measure $R_{I}(\rho)$ quantifies the output randomness in a measurement, our result can also be applied in other randomness generation scenarios \cite{colbeck2009quantum, gabriel2010generator, xu2012ultrafast, vazirani2011certifiable}.

%{We also investigate intrinsic randomness without specifying a measurement basis. In this case, we consider the scenario that an optimal measurement basis is chosen to maximize the output randomness. Here, we do not assume that the choice of measurement basis is secret from Eve¡¯s point of view. Thus, we still have to consider the minimization of intrinsic randomness on the chosen measurement basis. In this case, this basis-independent intrinsic randomness can be defined as $R(\rho) = \max_I R_I\left(\rho\right)$. In the qubit example, we show that the basis-independent intrinsic randomness is related to the purity of a quantum state. We thus demonstrate that intrinsic randomness can be used to quantify other quantum features.}
\acknowledgments
The author acknowledges insightful discussions with H.-K.~Lo. This work was supported by the National Basic Research Program of China Grants No.~2011CBA00300 and No.~2011CBA00301, and the 1000 Youth Fellowship program in China.

\appendix
\section{Verifying the properties of $R$}\label{Appendix:proof}
The requirements of coherence measures are listed in Table~\ref{Fig:Properties}. The intrinsic randomness measure is defined by
\begin{equation}\label{Eq:generaldef}
  R_I\left(\rho\right) = \min_{\{p_e,\ket{\psi_e}\}}\sum_e p_e R_I(\ket{\psi_e}),
\end{equation}
where the minimum runs over all possible decompositions of $\rho$, $\rho = \sum_e p_e \ket{\psi_e}\bra{\psi_e}$ and $\sum_{e}p_e= 1$.

In this section, we will show that the intrinsic randomness measure $R_I\left(\rho\right)$ satisfies the requirements of coherence measures. Here, we only show how to prove (C1') and (C2b), the proofs for the other requirements can be found in the main text.

\subsection{Proof of  (C1')}
To prove that $R_I(\rho)$ satisfies (C1'), consider a state $\rho\notin \mathcal{I}$ that has $R_I\left(\rho\right) = 0$. From the definition of $R_I$, there exists a decomposition $\rho = \sum_e p_e \ket{\psi_e}\bra{\psi_e}$ such that $R_I(\ket{\psi_e}\bra{\psi_e}) = 0$ for all $e$. Since any pure state with zero randomness is in the basis $I$, we have $\ket{\psi_e} = \ket{i_e} \in I$ and  $\rho = \sum_e p_e \ket{i_e}\bra{i_e}$, which belongs to the set $\mathcal{I}$. This leads to a contradiction.

\subsection{Proof of (C2b)}
As mentioned in the main text, the monotonicity requirement of (C2b) is satisfied for pure state,
\begin{equation}\label{Eq:Monotonicitypure}
  R_I\left(\ket{\psi}\right) \geq \sum_n p_nR_I\left(\ket{\psi_n}\right),
\end{equation}
where $\ket{\psi_n} = K_n\ket{\psi}/\sqrt{p_n}$, and $p_n = \mathrm{Tr}\left[K_n\ket{\psi}\bra{\psi}\right]$. This is because for   a pure state $\rho$,  the intrinsic randomness $R_I(\rho)$ equals the relative entropy coherence measure $C_{\mathrm{rel, ent}}(\rho)$ \cite{Baumgratz14}, whose monotonicity has already been proved.

For a general mixed state $\rho$, suppose that the optimal decomposition that achieves the minimum in Eq.~\eqref{Eq:generaldef} is given by $\rho = \sum_e p_e \ket{\psi_e}\bra{\psi_e}$. Then, we have
\begin{equation}\label{}
  R_I\left(\rho\right) = \sum_e p_e R_I(\ket{\psi_e})
\end{equation}
Now suppose that the incoherent operation defined in the main text is acted on $\rho$. What we need to prove is that
\begin{equation}\label{}
\begin{aligned}
  R_I\left(\rho\right) \geq \sum_n p_nR_I(\rho_n).
\end{aligned}
\end{equation}
where $  \rho_n = {K_n\rho K_n^\dag}/{p_n}$ and $p_n = \mathrm{Tr}\left[ K_n\rho K_n^\dag\right]$.
As $\rho = \sum_e p_e \ket{\psi_e}\bra{\psi_e}$, we have
\begin{equation}\label{}
\begin{aligned}
  \rho_n &= \frac{K_n\rho K_n^\dag}{p_n} \\
  &= \sum_e \frac{p_e}{p_n} {K_n\ket{\psi_e}\bra{\psi_e} K_n^\dag} \\
  &= \sum_e \frac{p_e}{p_n}p_{en}\rho_{en}
\end{aligned}
\end{equation}
where, we denote $p_{en} = \mathrm{Tr}[K_n\ket{\psi_e}\bra{\psi_e} K_n^\dag]$ and $\rho_{en} = {K_n\ket{\psi_e}\bra{\psi_e} K_n^\dag}/{p_{en}}$, and we have $p_n = \sum_e p_ep_{en}$. Then, we can finish the proof
\begin{equation}\label{}
\begin{aligned}
  R_I\left(\rho\right) &= \sum_e p_e R_I(\ket{\psi_e}) \\
  &\geq \sum_{e} p_e \sum_np_{en}R_I(\rho_{en})\\
   &= \sum_{n} p_n\sum_e\frac{p_e p_{xn}}{p_n}R_I(\rho_{en})\\
  &\geq \sum_{n} p_nR_I\left(\sum_e\frac{p_e p_{en}}{p_n}\rho_{en}\right)\\
   &= \sum_n p_nR_I(\rho_n),
    \end{aligned}
\end{equation}
where the first inequality is based on the conclusion for pure states in Eq.~\eqref{Eq:Monotonicitypure} and the last inequality is due to the convexity of $R_I$.

\section{Coherence distillation procedure}\label{App:distill}
A coherence distillation procedure refers to a series of incoherent operations by which a large number of identical partly coherent states can be transformed into a smaller number of maximally coherent states. In this section, we introduce a coherence distillation procedure for pure qubit states. With $N$ copies of states $\ket{\psi} = \alpha \ket{0} + \beta \ket{1}$, we show that we can asymptotically obtain $l$ copies of $\ket{\Psi_2} = (\ket{0} + \ket{1})/\sqrt{2}$, where $l$ and $N$ satisfy $l/N \approx R_I(\ket{\psi})$. The derivation method can be generalized  to an  arbitrary dimension.

\subsection{Qubit distillation}
First we prepare $MN$ copies of a partially coherent qubit state which will be uniformly divided into $M$ groups. The initial state of each group can be expressed according to
\begin{equation}
\ket{\psi}^{\otimes N}=\left(\alpha \ket{0}+\beta \ket{1}\right)^{\otimes N}.
\end{equation}
A binomial expansion on the computational basis contains $N+1$ distinct coefficients $\beta^N, \alpha^{1}\beta^{N-1}, \dots, \alpha^N$. Thus we can divide the original $2^N$-dimensional Hilbert space into $N+1$ subspaces according to the coefficients. For the $k$th coefficient $\alpha^{N-k}\beta^k$, the corresponding $k$th subspace is a  $D_k = C_N^k$ dimensional Hilbert space, whose basis are denoted by
\begin{equation}
\alpha^{N-k}\beta^k:\left\{\ket{e_1^k}, \ket{e_2^k}, \cdots , \ket{e_{D_k}^k}\right\}.
\end{equation}
When considering the computational basis, $\ket{e_i^k}$ $(i=, 1, 2, \cdots, D_k)$ is an $N$-qubit basis with $(N-k)$ $\ket{0}$s and $k$ $\ket{1}$s.

Next, we perform a projection measurement on $\ket{\psi}^{\otimes N}$ to the subspaces. In our case, the projection operator that maps onto the $k$th subspace is given by
\begin{equation}
P_k=\ket{e_1^k}\bra{e_1^k}+\ket{e_2^k}\bra{e_2^k}+\cdots+\ket{e_{D_k}^k}\bra{e_{D_k}^k}.
\end{equation}
The probability of obtaining the $k$th outcome is
\begin{equation}
p_k=C_N^k\left|\alpha\right|^{2N-2k}\left|\beta\right|^{2k}.
\end{equation}
Note that as the coefficients for the expansion are the same, the post-selection of the $k$th outcome corresponds to a maximally coherent state $\ket{\Psi_{D_k}}$  of dimension $D_k$.

If $D_k = 2^r$, we can directly convert to $r$ copies of $\ket{\Psi_2}$ as desired.
Or, we can repeat this process $M$ times, and take the tensor product of the post selected state to obtain a maximally coherent state of dimension $D$,
 \begin{equation}\label{}
      \ket{\Psi_{D}} = \ket{\Psi_{D_{k_1}}} \ket{\Psi_{D_{k_2}}} \dots\ket{\Psi_{D_{k_M}}},
\end{equation}
where $k_j$ is the outcome of the $j$th measurement, and the total dimension is $D = D_{k_1}D_{k_2}\cdots D_{k_M}$.
The total dimension $D$ will lie between $2^r$ and $2^r(1+\epsilon)$ $(0<\epsilon<1)$ for some power $r$. It can be proved \cite{Bennett96} that as $M$ increases, $\epsilon$ will asymptotically approach $0$.

Therefore, we can perform a second projection measurement to the $2^r$-dimensional Hilbert subspace and directly get obtain a final state
\begin{equation}
\ket{\Psi_2}^{\otimes r}=\left(\frac{1}{\sqrt{2}}\left(\ket{0}+\ket{1}\right)\right)^{\otimes r}
\end{equation}
Using the above procedure, $NM$ copies of a partly coherent qubit state $\alpha \ket{0}+ \beta \ket{1}$ have been distilled into $r$ copies of maximally coherent state.

In the following, we will show that all the operations of the distillation protocol are incoherent operations. In addition, we will show that the number of distilled maximally coherent state $r$ and the number of initial qubit $MN$ satisfy the relation $NMR_I(\ket{\psi}) \approx r$.
\subsection{Incoherent operations}
As the only operations are the two projective measurements, we only need to prove the following lemma.
\begin{lemma}
Suppose an $n$-dimensional Hilbert space has a complete basis $I_n = \{\ket{1}, \ket{2},\cdots, \ket{n}\}$. A projection measurement that divides $I_n$ into its complementary subsets is an incoherent operation on the basis of $I_n$.
\end{lemma}
\begin{proof}
Suppose that the basis $I_n$ is divided into $m$ complementary subsets $I_{n_1}, I_{n_2}, \cdots, I_{n_m}$, such that $I_{n_\alpha}\cap I_{n_\beta}= \emptyset$, for all $\alpha \neq \beta \in \{1,2,\cdots,m\}$, and $I_n = I_{n_1}\cup I_{n_2}\cup \cdots\cup I_{n_m}$. Denote the projector that projects onto the $I_{n_\alpha}$ subspace by $P_\alpha$. Thus, we can show that the projection measurement is a set of Kraus operators $\{\hat{P}_\alpha\}$ that satisfy $\hat{P}_\alpha^\dag\hat{P}_\beta = \delta_{\alpha,\beta}\hat{P}_\alpha$ and $\sum_\alpha P_\alpha = I_n$.
To prove the projection measurement to be an incoherent operation, we additionally need to show that $\hat{P}_\alpha\mathcal{I}_n \hat{P}_\alpha^\dag \subset \mathcal{I}_n$, where $\mathcal{I}_n$ is the set of all incoherent states that can be represented by $\delta=\sum_{i=1}^n \delta_i \ket{i}\bra{i}$. As the definition of $P_\alpha$, we have
\begin{equation}\label{}
\begin{aligned}
  P_\alpha \ket{i} &= \delta(\ket{i}\in I_{n_\alpha})\ket{i},
\end{aligned}
\end{equation}
where $\delta(\ket{i}\in I_{n_\alpha})=1$ if $\ket{i} \in I_{n_\alpha}$ and $\delta(\ket{i}\in I_{n_\alpha})=0$ otherwise.
Thus, we can show that for an arbitrary state $\delta=\sum_{i=1}^d \delta_i \ket{i}\bra{i}\in \mathcal{I}_n$, we have
\begin{equation}
\begin{aligned}
\hat{P}_\alpha \delta \hat{P}_\alpha^\dag &=\hat{P}_\alpha\sum_{i=1}^n \delta_i \ket{a_i}\bra{a_i} \hat{P}_\alpha^\dag \\
&= \sum_{i=1}^n \delta_i\delta(\ket{i}\in I_{n_\alpha}) \ket{a_i}\bra{a_i}\in \mathcal{I}_n.
\end{aligned}
\end{equation}
\end{proof}
Therefore, we have proven that the operations in the distillation protocol are incoherent.

\subsection{Coherence loss}
To explain why we have $NMR_I(\ket{\psi}) \approx r$, we only need to consider the coherence loss during the distillation process.
The initial state in each group can be rewrite as
\begin{equation}
\ket{\psi}^{\otimes N}=\sum_{k=0}^{N} \sqrt{C_N^k} \alpha^{N-k}\beta^k \ket{\Psi_{D_k}},
\end{equation}
where $\ket{\Psi_{D_k}}$ is a maximally coherent state of dimension $D_k$.
Thus the density matrix of the initial state is
\begin{equation}
\rho=\sum_{k,k'} \sqrt{C_N^k} \sqrt{C_N^{k'}}\alpha^{N-k}\beta^k (\alpha^*)^{N-k'}(\beta^*)^{k'} \ket{\Psi_k} \bra{\Psi_{k'}}
\end{equation}

Because the coherence of $\rho$ is defined by the von Neumann entropy of its diagonal terms, we first look at $\rho^{diag}$. That is,
\begin{equation}\label{1}
\begin{aligned}
\rho^{diag}=&\sum_{i=1}^{2^N}\bra{e_i}\rho\ket{e_i}\ket{e_i}\bra{e_i}\\
=&\sum_{i=1}^{2^N}\sum_{k,k'}\sqrt{C_N^k} \sqrt{C_N^{k'}}\alpha^{N-k}\beta^k (\alpha^*)^{N-k'}(\beta^*)^{k'} \\
&\bra{e_i}{\Psi_k}\rangle \bra{\Psi_{k'}}{e_i}\rangle\ket{e_i}\bra{e_i}.
\end{aligned}
\end{equation}
Here, we can see that when $k\neq k'$, $\bra{e_i}{\Psi_k}\rangle \bra{\Psi_{k'}}{e_i}\rangle=0$. Therefore Eq.~\eqref{1} can be simplified as
\begin{equation}
\rho^{diag}=\sum_{k=0}^{N} C_N^k\left|\alpha\right|^{2N-2k}\left|\beta\right|^{2k}\left(\ket{e_i}\bra{e_i}\right)^{diag}= \sum_{k=0}^{N} p_k \rho_k^{diag}.
\end{equation}
Here, $\rho^{\mathrm{diag}}$ has the decomposition $\{p_k, \rho_k^{diag}\}$. Thus, we have
\begin{equation}\label{2}
S(\rho^{diag})= H(p_k)+ \sum_{k=0}^{N}p_k S(\rho_k^{diag}),
\end{equation}
where $S(\rho^{diag})$ is the von Neumann entropy of $\rho^{diag}$ and $H(p_k)$ is the Shannon entropy. Considering our coherence (intrinsic randomness) definition, Eq.~\eqref{2} is equivalent to
\begin{equation}
C(\rho)=H(p_k)+\sum_{k=0}^{N} p_kC(\rho_k),
\end{equation}
where $C(\rho)$ is the average initial coherence and $\sum_{k=0}^{N} p_kC(\rho_k)$ is the average coherence left after the first projection measurement.
Therefore, the coherence loss in the first operation is
\begin{equation}
H(p_k)=-\sum_{k=0}^{N}p_k \log_2{(p_k)}\leq \log_2{N}.
\end{equation}
The coherence loss for the second projection measurement can be easily estimated by $\log_2(1+\epsilon)\approx \epsilon$. Thus the total coherence loss has an upper bound given by
\begin{equation}
M\log_2N+\log_2(1+\epsilon),
\end{equation}
which is negligible relative to the initial coherence $MNC(\ket{\psi})$ when $M$ and $N$ are large.

\section{Qubit example}\label{app:qubit}
Here, we derive the intrinsic randomness formula of the qubit state. We denote the Pauli matrices by $\sigma_i, \sigma_x, \sigma_y, \sigma_z$. When measured in the $\sigma_z$ basis, the intrinsic randomness for a pure qubit state $\ket{\psi} = \alpha \ket{0}+ \beta \ket{1}$ is given by
\begin{equation}\label{}
  R_z(\ket{\psi}) = H(|\alpha|^2) = H(|\beta|^2),
\end{equation}
where $H(p) = p\log p + (1-p)\log(1-p)$. If we define $n_x = \bra{\psi}\sigma_x\ket{\psi} = \alpha^*\beta + \alpha\beta^*$ and $n_y = \bra{\psi}\sigma_y\ket{\psi} = -i\alpha^*\beta + i\alpha\beta^*$, then it is easy to check that
\begin{equation}\label{Eq:RZ}
R_{z}(\ket{\psi})= H\left(\frac{1+\sqrt{1-n_x^2 - n_y^2}}{2}\right).
\end{equation}

For a general mixed state $\rho$, we can follow the  method for  deriving the entanglement of formation \cite{Wootters98}. In this case, we need to first define $\ket{\tilde{\psi}} = \sigma_x\ket{\psi^*} = \beta^*\ket{0}+  \alpha^*\ket{1}$, and the coherent concurrence by
\begin{equation}\label{}
  C_z(\ket{\psi}) = |\langle\psi|\tilde{\psi}\rangle| = 2|\alpha\beta|.
\end{equation}
Then it is easy to check that
\begin{equation}\label{Eq:RZ}
R_{z}(\ket{\psi})= H\left(\frac{1+\sqrt{1-C_z^2}}{2}\right).
\end{equation}
The randomness $R_{I}\left(\rho\right)$ can be obtained according to Eq.~\eqref{Eq:RZ} by first calculating the coherent concurrence. Follow the method of deriving the entanglement of formation, the  $C_z$ value can be obtained by $ C_z =  |\sqrt{\eta_1}-\sqrt{\eta_2}|$, where $\eta_1$ and $\eta_2$ are the eigenvalues of the matrix $M = \rho \sigma_x \rho^* \sigma_x$. In the Bloch sphere representation, the value of $C_z$ of a quantum state $\rho = (\sigma_i + n_x \sigma_x + n_y\sigma_y +n_z \sigma_z )/2$ can be calculated by
\begin{equation}\label{Eq:concurrence}
  C_z = \sqrt{n_x^2 + n_y^2}.
\end{equation}

Compared to the $l_1$ norm coherence measure $C_{l_1}$ \cite{Baumgratz14}, which is defined by the sum of the off-diagonal elements
\begin{equation}\label{}
C_{l_1}(\rho) = \sum_{i\neq j}|\rho_{ij}|,
\end{equation}
one can easily check that $C_{l_1}(\rho)$ equals the concurrence $C_z$ for the qubit case. This is because
\begin{equation}\label{}
\begin{aligned}
C_{l_1}(\rho) &= |\bra{0}\rho\ket{1}| + |\bra{1}\rho\ket{0}| \\
&= |\frac{1}{2}(n_x - in_y)| + |\frac{1}{2}(n_x + in_y)|\\
&= \sqrt{n_x^2 + n_y^2}.
\end{aligned}
\end{equation}

\section{General definitions for coherence measure} \label{Appendix:general}
Generally, when considering the intrinsic randomness of multiple copies of $\rho$, we can define the average intrinsic randomness in a manner similar  to the definition of entanglement cost \cite{hayden2001asymptotic, Plbnio07, Horodecki09} by
\begin{equation}\label{Eq:Nrho}
  R_I^C(\rho) = \inf\left\{r:\lim_{N\rightarrow\infty}\left[\inf_{\Phi_{\mathrm{ICPTP}}}D\left(\rho^{\otimes N}, \Phi_{\mathrm{ICPTP}}\left(\ket{\Psi_{2^{rN}}}\right)\right)\right]=0\right\},
\end{equation}
where $D(\rho_1,\rho_2)$ is a suitable measure of distance, which, for instance, could be the trace norm. In this case, the intrinsic randomness is understood as the average coherence cost in preparing $\rho$. Compared to the definition of the regularized entanglement of formation \cite{hayden2001asymptotic}, we conjecture that  $R_I^C(\rho)$ equals the regularized intrinsic randomness measure,
\begin{equation}\label{Eq:RNrho}
  R_I^\infty(\rho) = \lim_{N\rightarrow\infty}\frac{ R_I\left(\rho^{\otimes N}\right)}{N}.
\end{equation}

In the other direction, we can apply intrinsic operations to transform $N$ non-maximally coherent copies of $\rho$ to $l$ maximally coherent state $\ket{\Psi_2}$.
Similarly, we can define the distillable coherence by the supremum of $l/N$ over all possible distillation protocols \cite{Rains98, Plbnio07, Horodecki09},
\begin{equation}\label{Eq:Nrho}
  R_I^D(\rho) = \sup\left\{l:\lim_{N\rightarrow\infty}\left[\inf_{\Phi_{\mathrm{ICPTP}}}D\left(\Phi_{\mathrm{ICPTP}}\left(\rho^{\otimes N}\right)-\ket{\Psi_{2^{lN}}}\right)\right]=0\right\}.
\end{equation}
This distillable coherence $R_I^D(\rho)$ can thus be considered as the amount of intrinsic randomness when a quantum extractor is performed before measurement, as shown in the main text. For a general reasonable regularized coherence measure $C_I^\infty(\rho)$ similar to Eq.~\eqref{Eq:RNrho}, we conjecture that the two measures $R_I^D$ and $R_I^C$ are equivalent for all possible distance measures. They serves as two extremal measures, such that $R^D(\rho) \le C^\infty(\rho) \le R^C(\rho)$ for all regularized coherence measures $C^\infty(\rho)$.

%such that, $R_I^D \leq C_I^\infty \leq R_I^C$ for all regularized $C_I^\infty$. %In addition, similar to entanglement measures \cite{vidal2000entanglement, Horodecki00, donald2002uniqueness}, we show in the next section that the coherence measure for pure states is unique under regularization

%%%%%%%%%%%%%%%%%%%%%%%%%%%%%%%%%%%%%%%%
% choose a style
%\bibliographystyle{ieeetr}
%\bibliographystyle{unsrt}
\bibliographystyle{apsrev4-1}
%\bibliographystyle{naturemag}
%%%%%%%%%%%%%%%%%%%%%%%%%%%%%%%%%%%%%%%%

%%%%%%%%%%%%%%%%%%%%%%%%%%%%%%%%%%%%%%%%%
%% choose a .bib file
\bibliography{QR}

\begin{thebibliography}{45}%
\makeatletter
\providecommand \@ifxundefined [1]{%
 \@ifx{#1\undefined}
}%
\providecommand \@ifnum [1]{%
 \ifnum #1\expandafter \@firstoftwo
 \else \expandafter \@secondoftwo
 \fi
}%
\providecommand \@ifx [1]{%
 \ifx #1\expandafter \@firstoftwo
 \else \expandafter \@secondoftwo
 \fi
}%
\providecommand \natexlab [1]{#1}%
\providecommand \enquote  [1]{``#1''}%
\providecommand \bibnamefont  [1]{#1}%
\providecommand \bibfnamefont [1]{#1}%
\providecommand \citenamefont [1]{#1}%
\providecommand \href@noop [0]{\@secondoftwo}%
\providecommand \href [0]{\begingroup \@sanitize@url \@href}%
\providecommand \@href[1]{\@@startlink{#1}\@@href}%
\providecommand \@@href[1]{\endgroup#1\@@endlink}%
\providecommand \@sanitize@url [0]{\catcode `\\12\catcode `\$12\catcode
  `\&12\catcode `\#12\catcode `\^12\catcode `\_12\catcode `\%12\relax}%
\providecommand \@@startlink[1]{}%
\providecommand \@@endlink[0]{}%
\providecommand \url  [0]{\begingroup\@sanitize@url \@url }%
\providecommand \@url [1]{\endgroup\@href {#1}{\urlprefix }}%
\providecommand \urlprefix  [0]{URL }%
\providecommand \Eprint [0]{\href }%
\providecommand \doibase [0]{http://dx.doi.org/}%
\providecommand \selectlanguage [0]{\@gobble}%
\providecommand \bibinfo  [0]{\@secondoftwo}%
\providecommand \bibfield  [0]{\@secondoftwo}%
\providecommand \translation [1]{[#1]}%
\providecommand \BibitemOpen [0]{}%
\providecommand \bibitemStop [0]{}%
\providecommand \bibitemNoStop [0]{.\EOS\space}%
\providecommand \EOS [0]{\spacefactor3000\relax}%
\providecommand \BibitemShut  [1]{\csname bibitem#1\endcsname}%
\let\auto@bib@innerbib\@empty
%</preamble>
\bibitem [{\citenamefont {Born}(1926)}]{born1926quantentheorie}%
  \BibitemOpen
  \bibfield  {author} {\bibinfo {author} {\bibfnamefont {M.}~\bibnamefont
  {Born}},\ }\href@noop {} {\bibfield  {journal} {\bibinfo  {journal}
  {Zeitschrift f{\"u}r Physik}\ }\textbf {\bibinfo {volume} {37}},\ \bibinfo
  {pages} {863} (\bibinfo {year} {1926})}\BibitemShut {NoStop}%
\bibitem [{\citenamefont {Bell}(1987)}]{bell}%
  \BibitemOpen
  \bibfield  {author} {\bibinfo {author} {\bibfnamefont {J.~S.}\ \bibnamefont
  {Bell}},\ }\href@noop {} {\emph {\bibinfo {title} {On the
  Einstein-Podolsky-Rosen Paradox. Physics 1, 195--200 (1964)}}},\ Speakable
  and Unspeakable in Quantum Mechanics\ (\bibinfo  {publisher} {Cambridge
  University Press},\ \bibinfo {year} {1987})\BibitemShut {NoStop}%
\bibitem [{\citenamefont {Giovannetti}\ \emph {et~al.}(2011)\citenamefont
  {Giovannetti}, \citenamefont {Lloyd},\ and\ \citenamefont
  {Maccone}}]{giovannetti2011advances}%
  \BibitemOpen
  \bibfield  {author} {\bibinfo {author} {\bibfnamefont {V.}~\bibnamefont
  {Giovannetti}}, \bibinfo {author} {\bibfnamefont {S.}~\bibnamefont {Lloyd}},
  \ and\ \bibinfo {author} {\bibfnamefont {L.}~\bibnamefont {Maccone}},\
  }\href@noop {} {\bibfield  {journal} {\bibinfo  {journal} {Nature Photonics}\
  }\textbf {\bibinfo {volume} {5}},\ \bibinfo {pages} {222} (\bibinfo {year}
  {2011})}\BibitemShut {NoStop}%
\bibitem [{\citenamefont {Lambert}\ \emph {et~al.}(2013)\citenamefont
  {Lambert}, \citenamefont {Chen}, \citenamefont {Cheng}, \citenamefont {Li},
  \citenamefont {Chen},\ and\ \citenamefont {Nori}}]{lambert2013quantum}%
  \BibitemOpen
  \bibfield  {author} {\bibinfo {author} {\bibfnamefont {N.}~\bibnamefont
  {Lambert}}, \bibinfo {author} {\bibfnamefont {Y.-N.}\ \bibnamefont {Chen}},
  \bibinfo {author} {\bibfnamefont {Y.-C.}\ \bibnamefont {Cheng}}, \bibinfo
  {author} {\bibfnamefont {C.-M.}\ \bibnamefont {Li}}, \bibinfo {author}
  {\bibfnamefont {G.-Y.}\ \bibnamefont {Chen}}, \ and\ \bibinfo {author}
  {\bibfnamefont {F.}~\bibnamefont {Nori}},\ }\href@noop {} {\bibfield
  {journal} {\bibinfo  {journal} {Nature Physics}\ }\textbf {\bibinfo {volume}
  {9}},\ \bibinfo {pages} {10} (\bibinfo {year} {2013})}\BibitemShut {NoStop}%
\bibitem [{\citenamefont {Glauber}(1963)}]{Glauber63}%
  \BibitemOpen
  \bibfield  {author} {\bibinfo {author} {\bibfnamefont {R.}~\bibnamefont
  {Glauber}},\ }\href {\doibase 10.1103/PhysRev.131.2766} {\bibfield  {journal}
  {\bibinfo  {journal} {Phys. Rev.}\ }\textbf {\bibinfo {volume} {131}},\
  \bibinfo {pages} {2766} (\bibinfo {year} {1963})}\BibitemShut {NoStop}%
\bibitem [{\citenamefont {Sudarshan}(1963)}]{Sudarshan}%
  \BibitemOpen
  \bibfield  {author} {\bibinfo {author} {\bibfnamefont {E.}~\bibnamefont
  {Sudarshan}},\ }\href {\doibase 10.1103/PhysRevLett.10.277} {\bibfield
  {journal} {\bibinfo  {journal} {Phys. Rev. Lett.}\ }\textbf {\bibinfo
  {volume} {10}},\ \bibinfo {pages} {277} (\bibinfo {year} {1963})}\BibitemShut
  {NoStop}%
\bibitem [{\citenamefont {Luo}(2005)}]{luo2005quantum}%
  \BibitemOpen
  \bibfield  {author} {\bibinfo {author} {\bibfnamefont {S.}~\bibnamefont
  {Luo}},\ }\href@noop {} {\bibfield  {journal} {\bibinfo  {journal}
  {Theoretical and mathematical physics}\ }\textbf {\bibinfo {volume} {143}},\
  \bibinfo {pages} {681} (\bibinfo {year} {2005})}\BibitemShut {NoStop}%
\bibitem [{\citenamefont {{\AA{}berg}}(2006)}]{Aberg06}%
  \BibitemOpen
  \bibfield  {author} {\bibinfo {author} {\bibfnamefont {J.}~\bibnamefont
  {{\AA{}berg}}},\ }\href@noop {} {\bibfield  {journal} {\bibinfo  {journal}
  {eprint arXiv:quant-ph/0612146}\ } (\bibinfo {year} {2006})},\ \Eprint
  {http://arxiv.org/abs/quant-ph/0612146} {quant-ph/0612146} \BibitemShut
  {NoStop}%
\bibitem [{\citenamefont {Monras}\ \emph {et~al.}(2014)\citenamefont {Monras},
  \citenamefont {Checi\'nska},\ and\ \citenamefont
  {Ekert}}]{monras2013witnessing}%
  \BibitemOpen
  \bibfield  {author} {\bibinfo {author} {\bibfnamefont {A.}~\bibnamefont
  {Monras}}, \bibinfo {author} {\bibfnamefont {A.}~\bibnamefont {Checi\'nska}},
  \ and\ \bibinfo {author} {\bibfnamefont {A.}~\bibnamefont {Ekert}},\
  }\href@noop {} {\bibfield  {journal} {\bibinfo  {journal} {New Journal of
  Physics}\ }\textbf {\bibinfo {volume} {16}},\ \bibinfo {pages} {063041}
  (\bibinfo {year} {2014})}\BibitemShut {NoStop}%
\bibitem [{\citenamefont {Baumgratz}\ \emph {et~al.}(2014)\citenamefont
  {Baumgratz}, \citenamefont {Cramer},\ and\ \citenamefont
  {Plenio}}]{Baumgratz14}%
  \BibitemOpen
  \bibfield  {author} {\bibinfo {author} {\bibfnamefont {T.}~\bibnamefont
  {Baumgratz}}, \bibinfo {author} {\bibfnamefont {M.}~\bibnamefont {Cramer}}, \
  and\ \bibinfo {author} {\bibfnamefont {M.~B.}\ \bibnamefont {Plenio}},\
  }\href {\doibase 10.1103/PhysRevLett.113.140401} {\bibfield  {journal}
  {\bibinfo  {journal} {Phys. Rev. Lett.}\ }\textbf {\bibinfo {volume} {113}},\
  \bibinfo {pages} {140401} (\bibinfo {year} {2014})}\BibitemShut {NoStop}%
\bibitem [{\citenamefont {\AA{}berg}(2014)}]{Aberg14}%
  \BibitemOpen
  \bibfield  {author} {\bibinfo {author} {\bibfnamefont {J.}~\bibnamefont
  {\AA{}berg}},\ }\href {\doibase 10.1103/PhysRevLett.113.150402} {\bibfield
  {journal} {\bibinfo  {journal} {Phys. Rev. Lett.}\ }\textbf {\bibinfo
  {volume} {113}},\ \bibinfo {pages} {150402} (\bibinfo {year}
  {2014})}\BibitemShut {NoStop}%
\bibitem [{\citenamefont {Girolami}(2014)}]{Girolami14}%
  \BibitemOpen
  \bibfield  {author} {\bibinfo {author} {\bibfnamefont {D.}~\bibnamefont
  {Girolami}},\ }\href {\doibase 10.1103/PhysRevLett.113.170401} {\bibfield
  {journal} {\bibinfo  {journal} {Phys. Rev. Lett.}\ }\textbf {\bibinfo
  {volume} {113}},\ \bibinfo {pages} {170401} (\bibinfo {year}
  {2014})}\BibitemShut {NoStop}%
\bibitem [{\citenamefont {Bennett}\ \emph {et~al.}(1996)\citenamefont
  {Bennett}, \citenamefont {Bernstein}, \citenamefont {Popescu},\ and\
  \citenamefont {Schumacher}}]{Bennett96}%
  \BibitemOpen
  \bibfield  {author} {\bibinfo {author} {\bibfnamefont {C.~H.}\ \bibnamefont
  {Bennett}}, \bibinfo {author} {\bibfnamefont {H.~J.}\ \bibnamefont
  {Bernstein}}, \bibinfo {author} {\bibfnamefont {S.}~\bibnamefont {Popescu}},
  \ and\ \bibinfo {author} {\bibfnamefont {B.}~\bibnamefont {Schumacher}},\
  }\href {\doibase 10.1103/PhysRevA.53.2046} {\bibfield  {journal} {\bibinfo
  {journal} {Phys. Rev. A}\ }\textbf {\bibinfo {volume} {53}},\ \bibinfo
  {pages} {2046} (\bibinfo {year} {1996})}\BibitemShut {NoStop}%
\bibitem [{\citenamefont {Vedral}\ and\ \citenamefont
  {Plenio}(1998)}]{Vedral98}%
  \BibitemOpen
  \bibfield  {author} {\bibinfo {author} {\bibfnamefont {V.}~\bibnamefont
  {Vedral}}\ and\ \bibinfo {author} {\bibfnamefont {M.~B.}\ \bibnamefont
  {Plenio}},\ }\href {\doibase 10.1103/PhysRevA.57.1619} {\bibfield  {journal}
  {\bibinfo  {journal} {Phys. Rev. A}\ }\textbf {\bibinfo {volume} {57}},\
  \bibinfo {pages} {1619} (\bibinfo {year} {1998})}\BibitemShut {NoStop}%
\bibitem [{\citenamefont {Plenio}\ and\ \citenamefont
  {Virmani}(2007)}]{Plbnio07}%
  \BibitemOpen
  \bibfield  {author} {\bibinfo {author} {\bibfnamefont {M.~B.}\ \bibnamefont
  {Plenio}}\ and\ \bibinfo {author} {\bibfnamefont {S.}~\bibnamefont
  {Virmani}},\ }\href {http://dl.acm.org/citation.cfm?id=2011706.2011707}
  {\bibfield  {journal} {\bibinfo  {journal} {Quantum Info. Comput.}\ }\textbf
  {\bibinfo {volume} {7}},\ \bibinfo {pages} {1} (\bibinfo {year}
  {2007})}\BibitemShut {NoStop}%
\bibitem [{\citenamefont {Horodecki}\ \emph {et~al.}(2009)\citenamefont
  {Horodecki}, \citenamefont {Horodecki}, \citenamefont {Horodecki},\ and\
  \citenamefont {Horodecki}}]{Horodecki09}%
  \BibitemOpen
  \bibfield  {author} {\bibinfo {author} {\bibfnamefont {R.}~\bibnamefont
  {Horodecki}}, \bibinfo {author} {\bibfnamefont {P.}~\bibnamefont
  {Horodecki}}, \bibinfo {author} {\bibfnamefont {M.}~\bibnamefont
  {Horodecki}}, \ and\ \bibinfo {author} {\bibfnamefont {K.}~\bibnamefont
  {Horodecki}},\ }\href {\doibase 10.1103/RevModPhys.81.865} {\bibfield
  {journal} {\bibinfo  {journal} {Rev. Mod. Phys.}\ }\textbf {\bibinfo {volume}
  {81}},\ \bibinfo {pages} {865} (\bibinfo {year} {2009})}\BibitemShut
  {NoStop}%
\bibitem [{\citenamefont {{Zurek}}(2009)}]{Zurek09}%
  \BibitemOpen
  \bibfield  {author} {\bibinfo {author} {\bibfnamefont {W.~H.}\ \bibnamefont
  {{Zurek}}},\ }\href {\doibase 10.1038/nphys1202} {\bibfield  {journal}
  {\bibinfo  {journal} {Nature Physics}\ }\textbf {\bibinfo {volume} {5}},\
  \bibinfo {pages} {181} (\bibinfo {year} {2009})},\ \Eprint
  {http://arxiv.org/abs/0903.5082} {arXiv:0903.5082 [quant-ph]} \BibitemShut
  {NoStop}%
\bibitem [{\citenamefont {Hill}\ and\ \citenamefont {Wootters}(1997)}]{Hill97}%
  \BibitemOpen
  \bibfield  {author} {\bibinfo {author} {\bibfnamefont {S.}~\bibnamefont
  {Hill}}\ and\ \bibinfo {author} {\bibfnamefont {W.~K.}\ \bibnamefont
  {Wootters}},\ }\href {\doibase 10.1103/PhysRevLett.78.5022} {\bibfield
  {journal} {\bibinfo  {journal} {Phys. Rev. Lett.}\ }\textbf {\bibinfo
  {volume} {78}},\ \bibinfo {pages} {5022} (\bibinfo {year}
  {1997})}\BibitemShut {NoStop}%
\bibitem [{\citenamefont {Wootters}(1998)}]{Wootters98}%
  \BibitemOpen
  \bibfield  {author} {\bibinfo {author} {\bibfnamefont {W.~K.}\ \bibnamefont
  {Wootters}},\ }\href {\doibase 10.1103/PhysRevLett.80.2245} {\bibfield
  {journal} {\bibinfo  {journal} {Phys. Rev. Lett.}\ }\textbf {\bibinfo
  {volume} {80}},\ \bibinfo {pages} {2245} (\bibinfo {year}
  {1998})}\BibitemShut {NoStop}%
\bibitem [{Note1()}]{Note1}%
  \BibitemOpen
  \bibinfo {note} {The definitions of coherence and intrinsic randomness are
  based on a specific computational basis. In this perspective, the quantum
  feature can be quantified by the superposition strength on the measurement
  basis. Alternatively, we can define similar quantumness as the ability of
  measurements. For an arbitrary pure quantum state, if we can choose the
  measurement basis that is complementary to the state, a quantum feature
  similar to coherence can also be maximally revealed. The definitions of
  coherence based on the property of quantum state with a given measurement
  basis and the property of measurement is similar to the relationship between
  the pictures of Schrodinger and Heisenberg. The current definition of
  coherence thus follows from the routine of the Schrodinger¡¯s
  picture.}\BibitemShut {Stop}%
\bibitem [{\citenamefont {Shannon}(1948)}]{shannon2001mathematical}%
  \BibitemOpen
  \bibfield  {author} {\bibinfo {author} {\bibfnamefont {C.}~\bibnamefont
  {Shannon}},\ }\href {\doibase 10.1002/j.1538-7305.1948.tb01338.x} {\bibfield
  {journal} {\bibinfo  {journal} {Bell System Technical Journal, The}\ }\textbf
  {\bibinfo {volume} {27}},\ \bibinfo {pages} {379} (\bibinfo {year}
  {1948})}\BibitemShut {NoStop}%
\bibitem [{\citenamefont {R$\mathrm{\acute{e}}$nyi}(1961)}]{rrnyi1961measures}%
  \BibitemOpen
  \bibfield  {author} {\bibinfo {author} {\bibfnamefont {A.}~\bibnamefont
  {R$\mathrm{\acute{e}}$nyi}},\ }in\ \href@noop {} {\emph {\bibinfo {booktitle}
  {Fourth Berkeley Symposium on Mathematical Statistics and Probability}}}\
  (\bibinfo {year} {1961})\ pp.\ \bibinfo {pages} {547--561}\BibitemShut
  {NoStop}%
\bibitem [{\citenamefont {Hayden}\ \emph {et~al.}(2001)\citenamefont {Hayden},
  \citenamefont {Horodecki},\ and\ \citenamefont
  {Terhal}}]{hayden2001asymptotic}%
  \BibitemOpen
  \bibfield  {author} {\bibinfo {author} {\bibfnamefont {P.~M.}\ \bibnamefont
  {Hayden}}, \bibinfo {author} {\bibfnamefont {M.}~\bibnamefont {Horodecki}}, \
  and\ \bibinfo {author} {\bibfnamefont {B.~M.}\ \bibnamefont {Terhal}},\
  }\href@noop {} {\bibfield  {journal} {\bibinfo  {journal} {Journal of Physics
  A: Mathematical and General}\ }\textbf {\bibinfo {volume} {34}},\ \bibinfo
  {pages} {6891} (\bibinfo {year} {2001})}\BibitemShut {NoStop}%
\bibitem [{\citenamefont {Rungta}\ \emph {et~al.}(2001)\citenamefont {Rungta},
  \citenamefont {Bu\ifmmode~\check{z}\else \v{z}\fi{}ek}, \citenamefont
  {Caves}, \citenamefont {Hillery},\ and\ \citenamefont {Milburn}}]{Rungta01}%
  \BibitemOpen
  \bibfield  {author} {\bibinfo {author} {\bibfnamefont {P.}~\bibnamefont
  {Rungta}}, \bibinfo {author} {\bibfnamefont {V.}~\bibnamefont
  {Bu\ifmmode~\check{z}\else \v{z}\fi{}ek}}, \bibinfo {author} {\bibfnamefont
  {C.~M.}\ \bibnamefont {Caves}}, \bibinfo {author} {\bibfnamefont
  {M.}~\bibnamefont {Hillery}}, \ and\ \bibinfo {author} {\bibfnamefont
  {G.~J.}\ \bibnamefont {Milburn}},\ }\href {\doibase
  10.1103/PhysRevA.64.042315} {\bibfield  {journal} {\bibinfo  {journal} {Phys.
  Rev. A}\ }\textbf {\bibinfo {volume} {64}},\ \bibinfo {pages} {042315}
  (\bibinfo {year} {2001})}\BibitemShut {NoStop}%
\bibitem [{\citenamefont {Audenaert}\ \emph {et~al.}(2001)\citenamefont
  {Audenaert}, \citenamefont {Verstraete},\ and\ \citenamefont
  {De~Moor}}]{Audenaert01}%
  \BibitemOpen
  \bibfield  {author} {\bibinfo {author} {\bibfnamefont {K.}~\bibnamefont
  {Audenaert}}, \bibinfo {author} {\bibfnamefont {F.}~\bibnamefont
  {Verstraete}}, \ and\ \bibinfo {author} {\bibfnamefont {B.}~\bibnamefont
  {De~Moor}},\ }\href {\doibase 10.1103/PhysRevA.64.052304} {\bibfield
  {journal} {\bibinfo  {journal} {Phys. Rev. A}\ }\textbf {\bibinfo {volume}
  {64}},\ \bibinfo {pages} {052304} (\bibinfo {year} {2001})}\BibitemShut
  {NoStop}%
\bibitem [{\citenamefont {Badziag}\ \emph {et~al.}(2002)\citenamefont
  {Badziag}, \citenamefont {Deuar}, \citenamefont {Horodecki}, \citenamefont
  {Horodecki},\ and\ \citenamefont {Horodecki}}]{badziag2002concurrence}%
  \BibitemOpen
  \bibfield  {author} {\bibinfo {author} {\bibfnamefont {P.}~\bibnamefont
  {Badziag}}, \bibinfo {author} {\bibfnamefont {P.}~\bibnamefont {Deuar}},
  \bibinfo {author} {\bibfnamefont {M.}~\bibnamefont {Horodecki}}, \bibinfo
  {author} {\bibfnamefont {P.}~\bibnamefont {Horodecki}}, \ and\ \bibinfo
  {author} {\bibfnamefont {R.}~\bibnamefont {Horodecki}},\ }\href@noop {}
  {\bibfield  {journal} {\bibinfo  {journal} {Journal of Modern Optics}\
  }\textbf {\bibinfo {volume} {49}},\ \bibinfo {pages} {1289} (\bibinfo {year}
  {2002})}\BibitemShut {NoStop}%
\bibitem [{\citenamefont {Streltsov}\ \emph {et~al.}(2015)\citenamefont
  {Streltsov}, \citenamefont {Singh}, \citenamefont {Dhar}, \citenamefont
  {Bera},\ and\ \citenamefont {Adesso}}]{Streltsov15}%
  \BibitemOpen
  \bibfield  {author} {\bibinfo {author} {\bibfnamefont {A.}~\bibnamefont
  {Streltsov}}, \bibinfo {author} {\bibfnamefont {U.}~\bibnamefont {Singh}},
  \bibinfo {author} {\bibfnamefont {H.~S.}\ \bibnamefont {Dhar}}, \bibinfo
  {author} {\bibfnamefont {M.~N.}\ \bibnamefont {Bera}}, \ and\ \bibinfo
  {author} {\bibfnamefont {G.}~\bibnamefont {Adesso}},\ }\href {\doibase
  10.1103/PhysRevLett.115.020403} {\bibfield  {journal} {\bibinfo  {journal}
  {Phys. Rev. Lett.}\ }\textbf {\bibinfo {volume} {115}},\ \bibinfo {pages}
  {020403} (\bibinfo {year} {2015})}\BibitemShut {NoStop}%
\bibitem [{\citenamefont {Rains}(1999)}]{Rains98}%
  \BibitemOpen
  \bibfield  {author} {\bibinfo {author} {\bibfnamefont {E.~M.}\ \bibnamefont
  {Rains}},\ }\href {\doibase 10.1103/PhysRevA.60.173} {\bibfield  {journal}
  {\bibinfo  {journal} {Phys. Rev. A}\ }\textbf {\bibinfo {volume} {60}},\
  \bibinfo {pages} {173} (\bibinfo {year} {1999})}\BibitemShut {NoStop}%
\bibitem [{\citenamefont {Clauser}\ \emph {et~al.}(1969)\citenamefont
  {Clauser}, \citenamefont {Horne}, \citenamefont {Shimony},\ and\
  \citenamefont {Holt}}]{CHSH}%
  \BibitemOpen
  \bibfield  {author} {\bibinfo {author} {\bibfnamefont {J.~F.}\ \bibnamefont
  {Clauser}}, \bibinfo {author} {\bibfnamefont {M.~A.}\ \bibnamefont {Horne}},
  \bibinfo {author} {\bibfnamefont {A.}~\bibnamefont {Shimony}}, \ and\
  \bibinfo {author} {\bibfnamefont {R.~A.}\ \bibnamefont {Holt}},\ }\href
  {\doibase 10.1103/PhysRevLett.23.880} {\bibfield  {journal} {\bibinfo
  {journal} {Phys. Rev. Lett.}\ }\textbf {\bibinfo {volume} {23}},\ \bibinfo
  {pages} {880} (\bibinfo {year} {1969})}\BibitemShut {NoStop}%
\bibitem [{\citenamefont {{Du}}\ \emph {et~al.}(2015)\citenamefont {{Du}},
  \citenamefont {{Bai}},\ and\ \citenamefont {{Qi}}}]{Du15a}%
  \BibitemOpen
  \bibfield  {author} {\bibinfo {author} {\bibfnamefont {S.}~\bibnamefont
  {{Du}}}, \bibinfo {author} {\bibfnamefont {Z.}~\bibnamefont {{Bai}}}, \ and\
  \bibinfo {author} {\bibfnamefont {X.}~\bibnamefont {{Qi}}},\ }\href@noop {}
  {\bibfield  {journal} {\bibinfo  {journal} {ArXiv e-prints}\ } (\bibinfo
  {year} {2015})},\ \Eprint {http://arxiv.org/abs/1504.02862} {arXiv:1504.02862
  [quant-ph]} \BibitemShut {NoStop}%
\bibitem [{\citenamefont {Du}\ \emph {et~al.}(2015)\citenamefont {Du},
  \citenamefont {Bai},\ and\ \citenamefont {Guo}}]{Du15b}%
  \BibitemOpen
  \bibfield  {author} {\bibinfo {author} {\bibfnamefont {S.}~\bibnamefont
  {Du}}, \bibinfo {author} {\bibfnamefont {Z.}~\bibnamefont {Bai}}, \ and\
  \bibinfo {author} {\bibfnamefont {Y.}~\bibnamefont {Guo}},\ }\href {\doibase
  10.1103/PhysRevA.91.052120} {\bibfield  {journal} {\bibinfo  {journal} {Phys.
  Rev. A}\ }\textbf {\bibinfo {volume} {91}},\ \bibinfo {pages} {052120}
  (\bibinfo {year} {2015})}\BibitemShut {NoStop}%
\bibitem [{\citenamefont {Nielsen}(1999)}]{Nielsen99}%
  \BibitemOpen
  \bibfield  {author} {\bibinfo {author} {\bibfnamefont {M.~A.}\ \bibnamefont
  {Nielsen}},\ }\href {\doibase 10.1103/PhysRevLett.83.436} {\bibfield
  {journal} {\bibinfo  {journal} {Phys. Rev. Lett.}\ }\textbf {\bibinfo
  {volume} {83}},\ \bibinfo {pages} {436} (\bibinfo {year} {1999})}\BibitemShut
  {NoStop}%
\bibitem [{\citenamefont {D\"ur}\ \emph {et~al.}(2000)\citenamefont {D\"ur},
  \citenamefont {Vidal},\ and\ \citenamefont {Cirac}}]{Dur00}%
  \BibitemOpen
  \bibfield  {author} {\bibinfo {author} {\bibfnamefont {W.}~\bibnamefont
  {D\"ur}}, \bibinfo {author} {\bibfnamefont {G.}~\bibnamefont {Vidal}}, \ and\
  \bibinfo {author} {\bibfnamefont {J.~I.}\ \bibnamefont {Cirac}},\ }\href
  {\doibase 10.1103/PhysRevA.62.062314} {\bibfield  {journal} {\bibinfo
  {journal} {Phys. Rev. A}\ }\textbf {\bibinfo {volume} {62}},\ \bibinfo
  {pages} {062314} (\bibinfo {year} {2000})}\BibitemShut {NoStop}%
\bibitem [{\citenamefont {Jonathan}\ and\ \citenamefont
  {Plenio}(1999)}]{Jonathan99}%
  \BibitemOpen
  \bibfield  {author} {\bibinfo {author} {\bibfnamefont {D.}~\bibnamefont
  {Jonathan}}\ and\ \bibinfo {author} {\bibfnamefont {M.~B.}\ \bibnamefont
  {Plenio}},\ }\href {\doibase 10.1103/PhysRevLett.83.3566} {\bibfield
  {journal} {\bibinfo  {journal} {Phys. Rev. Lett.}\ }\textbf {\bibinfo
  {volume} {83}},\ \bibinfo {pages} {3566} (\bibinfo {year}
  {1999})}\BibitemShut {NoStop}%
\bibitem [{\citenamefont {Eisert}\ and\ \citenamefont
  {Wilkens}(2000)}]{Eisert00}%
  \BibitemOpen
  \bibfield  {author} {\bibinfo {author} {\bibfnamefont {J.}~\bibnamefont
  {Eisert}}\ and\ \bibinfo {author} {\bibfnamefont {M.}~\bibnamefont
  {Wilkens}},\ }\href {\doibase 10.1103/PhysRevLett.85.437} {\bibfield
  {journal} {\bibinfo  {journal} {Phys. Rev. Lett.}\ }\textbf {\bibinfo
  {volume} {85}},\ \bibinfo {pages} {437} (\bibinfo {year} {2000})}\BibitemShut
  {NoStop}%
\bibitem [{\citenamefont {Mayers}\ and\ \citenamefont {Yao}(1998)}]{Mayers98}%
  \BibitemOpen
  \bibfield  {author} {\bibinfo {author} {\bibfnamefont {D.}~\bibnamefont
  {Mayers}}\ and\ \bibinfo {author} {\bibfnamefont {A.}~\bibnamefont {Yao}},\
  }in\ \href {http://dl.acm.org/citation.cfm?id=795664.796390} {\emph {\bibinfo
  {booktitle} {Proceedings of the 39th Annual Symposium on Foundations of
  Computer Science}}},\ \bibinfo {series and number} {FOCS '98}\ (\bibinfo
  {publisher} {IEEE Computer Society},\ \bibinfo {address} {Washington, DC,
  USA},\ \bibinfo {year} {1998})\ pp.\ \bibinfo {pages} {503--}\BibitemShut
  {NoStop}%
\bibitem [{\citenamefont {Ac\'in}\ \emph {et~al.}(2006)\citenamefont {Ac\'in},
  \citenamefont {Gisin},\ and\ \citenamefont {Masanes}}]{Acin}%
  \BibitemOpen
  \bibfield  {author} {\bibinfo {author} {\bibfnamefont {A.}~\bibnamefont
  {Ac\'in}}, \bibinfo {author} {\bibfnamefont {N.}~\bibnamefont {Gisin}}, \
  and\ \bibinfo {author} {\bibfnamefont {L.}~\bibnamefont {Masanes}},\ }\href
  {\doibase 10.1103/PhysRevLett.97.120405} {\bibfield  {journal} {\bibinfo
  {journal} {Phys. Rev. Lett.}\ }\textbf {\bibinfo {volume} {97}},\ \bibinfo
  {pages} {120405} (\bibinfo {year} {2006})}\BibitemShut {NoStop}%
\bibitem [{\citenamefont {Vazirani}\ and\ \citenamefont
  {Vidick}(2012)}]{vazirani2011certifiable}%
  \BibitemOpen
  \bibfield  {author} {\bibinfo {author} {\bibfnamefont {U.}~\bibnamefont
  {Vazirani}}\ and\ \bibinfo {author} {\bibfnamefont {T.}~\bibnamefont
  {Vidick}},\ }in\ \href {\doibase 10.1145/2213977.2213984} {\emph {\bibinfo
  {booktitle} {Proceedings of the Forty-fourth Annual ACM Symposium on Theory
  of Computing}}},\ \bibinfo {series and number} {STOC '12}\ (\bibinfo
  {publisher} {ACM},\ \bibinfo {address} {New York, NY, USA},\ \bibinfo {year}
  {2012})\ pp.\ \bibinfo {pages} {61--76}\BibitemShut {NoStop}%
\bibitem [{\citenamefont {Lo}\ \emph {et~al.}(2012)\citenamefont {Lo},
  \citenamefont {Curty},\ and\ \citenamefont {Qi}}]{Lo12}%
  \BibitemOpen
  \bibfield  {author} {\bibinfo {author} {\bibfnamefont {H.-K.}\ \bibnamefont
  {Lo}}, \bibinfo {author} {\bibfnamefont {M.}~\bibnamefont {Curty}}, \ and\
  \bibinfo {author} {\bibfnamefont {B.}~\bibnamefont {Qi}},\ }\href {\doibase
  10.1103/PhysRevLett.108.130503} {\bibfield  {journal} {\bibinfo  {journal}
  {Phys. Rev. Lett.}\ }\textbf {\bibinfo {volume} {108}},\ \bibinfo {pages}
  {130503} (\bibinfo {year} {2012})}\BibitemShut {NoStop}%
\bibitem [{\citenamefont {Braunstein}\ and\ \citenamefont
  {Pirandola}(2012)}]{Braunstein12}%
  \BibitemOpen
  \bibfield  {author} {\bibinfo {author} {\bibfnamefont {S.~L.}\ \bibnamefont
  {Braunstein}}\ and\ \bibinfo {author} {\bibfnamefont {S.}~\bibnamefont
  {Pirandola}},\ }\href {\doibase 10.1103/PhysRevLett.108.130502} {\bibfield
  {journal} {\bibinfo  {journal} {Phys. Rev. Lett.}\ }\textbf {\bibinfo
  {volume} {108}},\ \bibinfo {pages} {130502} (\bibinfo {year}
  {2012})}\BibitemShut {NoStop}%
\bibitem [{\citenamefont {Branciard}\ \emph {et~al.}(2013)\citenamefont
  {Branciard}, \citenamefont {Rosset}, \citenamefont {Liang},\ and\
  \citenamefont {Gisin}}]{Branciard13}%
  \BibitemOpen
  \bibfield  {author} {\bibinfo {author} {\bibfnamefont {C.}~\bibnamefont
  {Branciard}}, \bibinfo {author} {\bibfnamefont {D.}~\bibnamefont {Rosset}},
  \bibinfo {author} {\bibfnamefont {Y.-C.}\ \bibnamefont {Liang}}, \ and\
  \bibinfo {author} {\bibfnamefont {N.}~\bibnamefont {Gisin}},\ }\href
  {\doibase 10.1103/PhysRevLett.110.060405} {\bibfield  {journal} {\bibinfo
  {journal} {Phys. Rev. Lett.}\ }\textbf {\bibinfo {volume} {110}},\ \bibinfo
  {pages} {060405} (\bibinfo {year} {2013})}\BibitemShut {NoStop}%
\bibitem [{\citenamefont {Xu}\ \emph {et~al.}(2014)\citenamefont {Xu},
  \citenamefont {Yuan}, \citenamefont {Chen}, \citenamefont {Lu}, \citenamefont
  {Yao}, \citenamefont {Ma}, \citenamefont {Chen},\ and\ \citenamefont
  {Pan}}]{Xu14}%
  \BibitemOpen
  \bibfield  {author} {\bibinfo {author} {\bibfnamefont {P.}~\bibnamefont
  {Xu}}, \bibinfo {author} {\bibfnamefont {X.}~\bibnamefont {Yuan}}, \bibinfo
  {author} {\bibfnamefont {L.-K.}\ \bibnamefont {Chen}}, \bibinfo {author}
  {\bibfnamefont {H.}~\bibnamefont {Lu}}, \bibinfo {author} {\bibfnamefont
  {X.-C.}\ \bibnamefont {Yao}}, \bibinfo {author} {\bibfnamefont
  {X.}~\bibnamefont {Ma}}, \bibinfo {author} {\bibfnamefont {Y.-A.}\
  \bibnamefont {Chen}}, \ and\ \bibinfo {author} {\bibfnamefont {J.-W.}\
  \bibnamefont {Pan}},\ }\href {\doibase 10.1103/PhysRevLett.112.140506}
  {\bibfield  {journal} {\bibinfo  {journal} {Phys. Rev. Lett.}\ }\textbf
  {\bibinfo {volume} {112}},\ \bibinfo {pages} {140506} (\bibinfo {year}
  {2014})}\BibitemShut {NoStop}%
\bibitem [{\citenamefont {Colbeck}(2009)}]{colbeck2009quantum}%
  \BibitemOpen
  \bibfield  {author} {\bibinfo {author} {\bibfnamefont {R.}~\bibnamefont
  {Colbeck}},\ }\href@noop {} {\bibfield  {journal} {\bibinfo  {journal} {arXiv
  preprint arXiv:0911.3814}\ } (\bibinfo {year} {2009})}\BibitemShut {NoStop}%
\bibitem [{\citenamefont {Gabriel}\ \emph {et~al.}(2010)\citenamefont
  {Gabriel}, \citenamefont {Wittmann}, \citenamefont {Sych}, \citenamefont
  {Dong}, \citenamefont {Mauerer}, \citenamefont {Andersen}, \citenamefont
  {Marquardt},\ and\ \citenamefont {Leuchs}}]{gabriel2010generator}%
  \BibitemOpen
  \bibfield  {author} {\bibinfo {author} {\bibfnamefont {C.}~\bibnamefont
  {Gabriel}}, \bibinfo {author} {\bibfnamefont {C.}~\bibnamefont {Wittmann}},
  \bibinfo {author} {\bibfnamefont {D.}~\bibnamefont {Sych}}, \bibinfo {author}
  {\bibfnamefont {R.}~\bibnamefont {Dong}}, \bibinfo {author} {\bibfnamefont
  {W.}~\bibnamefont {Mauerer}}, \bibinfo {author} {\bibfnamefont {U.~L.}\
  \bibnamefont {Andersen}}, \bibinfo {author} {\bibfnamefont {C.}~\bibnamefont
  {Marquardt}}, \ and\ \bibinfo {author} {\bibfnamefont {G.}~\bibnamefont
  {Leuchs}},\ }\href@noop {} {\bibfield  {journal} {\bibinfo  {journal} {Nature
  Photonics}\ }\textbf {\bibinfo {volume} {4}},\ \bibinfo {pages} {711}
  (\bibinfo {year} {2010})}\BibitemShut {NoStop}%
\bibitem [{\citenamefont {Xu}\ \emph {et~al.}(2012)\citenamefont {Xu},
  \citenamefont {Qi}, \citenamefont {Ma}, \citenamefont {Xu}, \citenamefont
  {Zheng},\ and\ \citenamefont {Lo}}]{xu2012ultrafast}%
  \BibitemOpen
  \bibfield  {author} {\bibinfo {author} {\bibfnamefont {F.}~\bibnamefont
  {Xu}}, \bibinfo {author} {\bibfnamefont {B.}~\bibnamefont {Qi}}, \bibinfo
  {author} {\bibfnamefont {X.}~\bibnamefont {Ma}}, \bibinfo {author}
  {\bibfnamefont {H.}~\bibnamefont {Xu}}, \bibinfo {author} {\bibfnamefont
  {H.}~\bibnamefont {Zheng}}, \ and\ \bibinfo {author} {\bibfnamefont {H.-K.}\
  \bibnamefont {Lo}},\ }\href@noop {} {\bibfield  {journal} {\bibinfo
  {journal} {Optics express}\ }\textbf {\bibinfo {volume} {20}},\ \bibinfo
  {pages} {12366} (\bibinfo {year} {2012})}\BibitemShut {NoStop}%
\end{thebibliography}%
%%%%%%%%%%%%%%%%%%%%%%%%%%%%%%%%%%%%%%%%%

%% Here is the endmatter stuff: Supplementary Info, etc.
%% Use \item's to separate, default label is "Acknowledgements"

%%
%% TABLES
%%
%% If there are any tables, put them here.
%%

\end{document}